\documentclass[twocolumn,notitlepage,longbibliography,floatfix]{revtex4-2}
\usepackage[utf8]{inputenc}
\usepackage{url}
\usepackage[noend]{algpseudocode}
\usepackage{algorithm}
\usepackage{verbatim}
\usepackage{amsmath,amssymb}
\usepackage{amsthm}
\usepackage{multirow}
\usepackage{physics}
\usepackage[pdftex]{graphicx}
\usepackage{xcolor}
\usepackage{array}
\usepackage{caption}
\usepackage{subcaption}
\usepackage{setspace}
\usepackage{tabto}
\usepackage{lineno}
\usepackage{bbold}
\usepackage[shortlabels]{enumitem}

\newtheorem{theorem}{Theorem}

\newtheorem{lemma}{Lemma}

\newcommand*\patchAmsMathEnvironmentForLineno[1]{%
      \expandafter\let\csname old#1\expandafter\endcsname\csname #1\endcsname
      \expandafter\let\csname oldend#1\expandafter\endcsname\csname end#1\endcsname
      \renewenvironment{#1}%
         {\linenomath\csname old#1\endcsname}%
         {\csname oldend#1\endcsname\endlinenomath}}%
    \newcommand*\patchBothAmsMathEnvironmentsForLineno[1]{%
      \patchAmsMathEnvironmentForLineno{#1}%
      \patchAmsMathEnvironmentForLineno{#1*}}%
    \AtBeginDocument{%
    \patchBothAmsMathEnvironmentsForLineno{equation}%
    \patchBothAmsMathEnvironmentsForLineno{align}%
    \patchBothAmsMathEnvironmentsForLineno{flalign}%
    \patchBothAmsMathEnvironmentsForLineno{alignat}%
    \patchBothAmsMathEnvironmentsForLineno{gather}%
    \patchBothAmsMathEnvironmentsForLineno{multline}%
    }
    \include{def}
\def\dispmuskip{\thinmuskip= 3mu plus 0mu minus 2mu \medmuskip=  4mu plus 2mu minus 2mu \thickmuskip=5mu plus 5mu minus 2mu}
\def\textmuskip{\thinmuskip= 0mu                    \medmuskip=  1mu plus 1mu minus 1mu \thickmuskip=2mu plus 3mu minus 1mu}
\def\beq{\dispmuskip\begin{equation}}    \def\eeq{\end{equation}\textmuskip}
\def\beqn{\dispmuskip\begin{displaymath}}\def\eeqn{\end{displaymath}\textmuskip}
\def\bea{\dispmuskip\begin{eqnarray}}    \def\eea{\end{eqnarray}\textmuskip}
\def\bean{\dispmuskip\begin{eqnarray*}}  \def\eean{\end{eqnarray*}\textmuskip}

\newcounter{mntcomm}

\newcounter{alcomm}

\newcommand{\wh}{\widehat}

\def\E{{\mathbb E}}                         % Expectation

\def\l{\lambda}

\def\M{{\cal M}}

\def\cov{\text{\rm Cov}}

\usepackage[margin=1in]{geometry}

\begin{document}

\title{Quantum Speedup of Natural Gradient for Variational Bayes}
\thanks{\textit{The research was partially supported by the Australian Research Council's Discovery Project DP200103015 and the Australian Centre of Excellence for Mathematical and Statistical Frontiers (ACEMS).}}

\author{Anna Lopatnikova}
 \email{anna.lopatnikova@sydney.edu.au}

\author{Minh-Ngoc Tran}
 \email{minh-ngoc.tran@sydney.edu.au}

\affiliation{%
 Discipline of Business Analytics, The University of Sydney Business School, The University of Sydney, NSW 2006, Australia
}%

\date{\today}

\begin{abstract}
    Variational Bayes (VB) is a critical method in machine learning and statistics, underpinning the recent success of Bayesian deep learning.  The natural gradient is an essential component of efficient VB estimation, but it is prohibitively computationally expensive in high dimensions.   We propose a computationally efficient regression-based method for natural gradient estimation, with convergence guarantees under standard assumptions.  The method enables the use of quantum matrix inversion to further speed up VB.  We demonstrate that the problem setup fulfills the conditions required for quantum matrix inversion to deliver computational efficiency. The method works with a broad range of statistical models and does not require special-purpose or simplified variational distributions.   \\
    
\end{abstract}

\maketitle

%\doublespacing

\section{Introduction}

Our work demonstrates how a general quantum computer could help unlock the next generation of deep learning applications.   We focus on Variational Bayes (VB), a critical method which has underpinned the recent surge in successful industrial and research applications of Bayesian deep learning; see \cite{Hoffman:2013,Kingma.Welling:VAE,zhang2018advances} and references therein.   Even though VB is efficient and scalable relative to alternative methods, it remains too computationally intensive for many practical applications, particularly when the number of variational parameters $N$ is large (e.g.~millions or more). 

We propose a quantum-classical algorithm to speed up VB through efficient computation of natural gradient.   The use of natural gradient is seen as one of the most promising ways to speed up VB \cite{amari1998natural,Sato:2001,Tran:2017,Mishkin:2018}.  However, this method requires the computation and inversion of the Fisher information matrix  \cite{martens2014new,martens2015,Tran2020JCGS}, which, for a model with a million parameters, would require a classical computer to perform, at best, quintillions of operations -- $\tilde{O}(N^d)$ with $2<d\leq 3$ (Section~\ref{sec:vb_with_ng}).

To make the problem of natural gradient estimation suitable for quantum speedup, we reformulate it into a linear regression problem that exploits the structure of the Fisher matrix (Section~\ref{sec:regression}).  We show that the resulting problem could be one of the best practical uses of the quantum linear systems algorithms \cite{harrow2009quantum,ambainis2012variable,clader2013preconditioned,kerenidis2016quantum,wossnig2018quantum,childs2017quantum,subacsi2019quantum,gilyen2019quantum}. These algorithms offer dramatic speedups but have stringent requirements, rarely met by real-world problems \cite{clader2013preconditioned,aaronson2015read}.  We show that the problem of natural gradient estimation for VB meets the critical requirements, enabling speedup to the classical limit of $O(N)$ needed to write down all the variational parameters (Section~\ref{sec:QVB}).  Additionally, the linear regression formulation of natural gradient estimation enables efficient classical computation with time complexity $\tilde{O}(N^b)$ with $1 \leq b < 2$ (Section~\ref{sec:efficiency_nat_grad}).  We demonstrate the method’s power in a classical simulation where the VB approximation is a computationally intensive neural network  (Section~\ref{sec: numerical examples}).

%====================================================%
\section{A Regression-Based Approach to Natural Gradient Estimation for Variational Bayes}\label{sec: classical VB}
%====================================================%

In this section, we provide a review of VB that uses natural gradient optimization, then demonstrate the regression-based approach and prove its convergence.

\subsection{Variational Bayes with Natural Gradient}
\label{sec:vb_with_ng}

One of the most promising classical advances in VB is the use of natural gradient instead of Euclidean gradient in stochastic gradient descent
\citep{amari1998natural,Sato:2001,Tran:2017,Mishkin:2018}.  Being a geometric object that takes into account the information geometry of the variational family, natural gradient often leads to faster and more stable convergence than alternative methods \cite{martens2014new}. In situations where 
variational parameters lie on Riemannian manifolds, natural gradient is the only reliable adaptive learning method that  speeds up VB training \cite{Tran:2019VB_manifold}.
However, in many cases, accurate computation of natural gradient is infeasible, because it requires the {\it analytic} computation and {\it inversion} of the Fisher information matrix.  With computational complexity of $\tilde{O}(N^d)$, with $2<d\leq 3$ depending on various algorithms ($d=3$ in practice), natural gradient estimation is prohibitively computationally expensive in high-dimensional cases; current practice resorts to heuristic workarounds that can affect the results of Bayesian inference \cite{martens2015}.  We propose a regression-based approach to natural gradient estimation for VB, which lends itself to efficient computational methods, including quantum computing algorithms (Section~\ref{sec:QVB}).

Let $y$ be the data and $L(\theta)=p(y|\theta)$ the likelihood function based on a statistical model, with $\theta$ the set of model parameters to be estimated.
Let $p(\theta)$ be the prior. 
Bayesian inference encodes all information into the posterior distribution
\[p(\theta|y)=\frac{p(\theta)L(\theta)}{p(y)}.\]
The main task in Bayesian inference is to approximate
the posterior $p(\theta|y)$, and VB does so by approximating it by a probability distribution with density $q_\lambda(\theta)$, $\lambda\in\M$ - the variational parameter space, belonging to some tractable family of distributions such as Gaussians or neural networks. Denote by $N$ the size of $\lambda$. The best $\lambda$ is found by maximizing the lower bound \citep{zhang2018advances}
\begin{equation}\label{eq: lower bound}
\mathcal L(\lambda)=\int q_\lambda(\theta)\log\frac{p(\theta)L(\theta)}{q_\lambda(\theta)}d\theta=\E_{q_\lambda}\big(h_\lambda(\theta)\big),
\end{equation}
with
\begin{equation}\label{eq: h_lambda function}
h_\lambda(\theta)=\log\big(p(\theta)L(\theta)\big)-\log q_\lambda(\theta).
\end{equation}
The Euclidean gradient of the lower bound is \citep{tran2021VBtutorial}
\begin{equation}\label{eq: LB gradient}
\nabla_\lambda\mathcal L(\lambda) = \E_{q_\lambda}\big(\nabla_\lambda\log q_\lambda(\theta)\times h_\lambda(\theta)\big).
\end{equation}
For maximizing the lower bound objective function $\mathcal L(\lambda)$, Amari \cite{amari1998natural} shows that the {\it natural gradient}, defined as
\begin{equation}\label{eq:natural gradient solution}
\nabla_{\lambda}^{\text{nat}}\mathcal L(\lambda):=I^{-1}_F(\lambda)\nabla_{\lambda}\mathcal L(\lambda),
\end{equation}
where 
\begin{equation}\label{eq:I_F}
I_F(\lambda):=\text{Cov}_{q_\lambda}\big(\nabla_\lambda\log q_\lambda(\theta),\nabla_\lambda\log q_\lambda(\theta)\big)
\end{equation}
is the Fisher information matrix of $q_\lambda$, works much more efficiently than the Euclidean gradient.
The reason is that the natural gradient takes into account the geometry of the variational family $\{q_\lambda, \lambda\in\M\}$ in the optimization.
Natural gradient has been proven vital to the success of VB; see, e.g., Refs.~\cite{Sato:2001,Hoffman:2013} among others.

%\subsubsection*{Classical VB algorithm}
The classical VB algorithm with the use of natural gradient proceeds as follows (Algorithm~\ref{Al:algorithm 1}). Let $\wh{\nabla_\l^\text{nat}{\mathcal L}}(\l)$ be an estimate of the natural gradient ${\nabla_\l^\text{nat}{\mathcal L}}(\l)$.  At step $t$ of stochastic natural gradient descent, the natural gradient estimate is used to compute the momentum gradient $Y_{t+1}$, which in turn updates the variational parameter $\l_{t+1}=\l_t+\alpha_tY_{t+1}$.  
\begin{algorithm}[H]
\caption{VB Algorithm with Natural Gradient\label{Al:algorithm 1}}
\vspace{-1.5ex}
\begin{enumerate}[leftmargin=*]
\itemsep0em 
  \item Initialize $\l_0$ and momentum gradient $Y_0$. 
  \item For $t=0,1,...$ and until stopping criterion is met:
\vspace{-1.5ex}
  \begin{enumerate}[a.,leftmargin=*]
  \itemsep0em 
	  %\item Generate $\theta_s\sim q_{\lambda_t}(\theta)$, $s=1,...,S$
	  \item Estimate the natural gradient $\wh{\nabla_\l^\text{nat}{\mathcal L}}(\l_t)$ and apply gradient clipping. 
	  %\[\wh{\nabla_\l{\mathcal L}}(\lambda_{t}):=\frac{1}{S}\sum_{s=1}^S\nabla_\lambda \log q_\lambda(\theta_s)\times h_\lambda(\theta_s)|_{\lambda=\lambda^{(t)}}.\]
	  \item Update the momentum gradient \label{step:momgrad}
	  \begin{align}
	      Y_{t+1}& =\omega Y_t+(1-\omega)\wh{\nabla_\l^\text{nat}{\mathcal L}}(\l_t).\label{eq:mom grad}
	  \end{align} 
	  \item Update the variational parameter
	  \bea\label{eq:update rule}
	  \lambda_{t+1}=\lambda_{t}+\alpha_t Y_{t+1}.
	  \eea
  \end{enumerate}
\end{enumerate}
\end{algorithm}
The momentum gradient method, found highly useful in practice, smooths out the natural gradient estimate as in (\ref{eq:mom grad}) incorporating the latest estimate $\wh{\nabla_\l^\text{nat}{\mathcal L}}(\l_t)$ with the momentum weight $\omega\in(0,1)$. The learning rates $\alpha_t$ are required to satisfy: $1\geq\alpha_t\downarrow 0$, $\sum \alpha_t=\infty$ and $\sum \alpha_t^2<\infty$. We refer the interested reader to Ref.~\cite{tran2021VBtutorial} for more details.   

The current practice of stochastic gradient descent for VB often includes
{\it gradient clipping} that normalizes the length of $\wh{\nabla_\l^\text{nat}{\mathcal L}}(\l_t)$ \cite{Goodfellow-et-al-2016}. 
%This normalization is automatically incorporated into our quantum natural gradient because of the unit norm requirement of quantum states.
%The scaling factor $\kappa$ is a tuning parameter.

The main computational bottleneck in Variational Bayes is computing the gradient estimate $\wh{\nabla_\l^\text{nat}{\mathcal L}}(\l)$. We develop the regression-based natural gradient estimation method to overcome this problem, including with the use of a hybrid quantum-classical algorithm.

\subsection{Estimating Natural Gradient with Regression}
\label{sec:regression}

The obstacles to practical implementation of natural gradient descent are well-known. It is difficult to compute the Fisher matrix let alone its inverse. Even if one could compute the Fisher matrix analytically, solving the equation in \eqref{eq:natural gradient solution} is computationally infeasible in high dimensions. Our regression-based approach to natural gradient estimation overcomes these difficulties. 
%We emphasize that this  regression-based natural gradient method itself is also useful for classical VB.

Denote $T(\theta)=\nabla_\lambda \log q_\lambda(\theta)$, a vector-valued function often called the \emph{score function}.  By noting that $\E_{q_\lambda}(T(\theta))=0$, the Euclidean gradient of the lower bound in \eqref{eq: LB gradient} and the Fisher matrix in \eqref{eq:I_F} can be rewritten as
\begin{align}
&\nabla_\lambda\mathcal L(\lambda) = \cov_{q_\lambda}\big(T(\theta),h_\lambda(\theta)\big),\;\;\;\nonumber\\ &\text{and}\;\;\;
I_F(\lambda) = \cov_{q_\lambda}\big(T(\theta),T(\theta)\big).
\end{align}
Hence, the natural gradient takes the form
\begin{align}
    &\nabla_\l^\text{nat}{\mathcal L}(\l)= I_F^{-1}(\lambda)\nabla_\lambda\mathcal{L}(\lambda)\nonumber \\ 
    &=  \big[\cov_{q_\lambda} (T(\theta),T(\theta))\big]^{-1} \cov_{q_\lambda}\big(T(\theta),h_\lambda(\theta)\big). \label{eq:natgrad_eta}
\end{align}

Equation \eqref{eq:natgrad_eta} invites the interpretation of the natural gradient as the slope vector $\beta$ of regression of the dependent variable $h_\lambda(\theta)$ on the vector  of independent variables $T(\theta)$.  

More concretely, let $\{\theta_i\}_{i=1}^M$ be $M$ samples from $q_\lambda(\theta)$, and consider the linear regression
\beq\label{eq:regression prob}
h_\lambda(\theta_i)=\beta_0+T(\theta_i)^\top\beta+\epsilon_i,
\eeq
where the $\epsilon_i$ are independently distributed with mean zero. We need an intercept $\beta_0$ in the regression model \eqref{eq:regression prob}  because $\E_{q_\l}(T(\theta))=0$. We note that $\beta_0=\E_{q_\lambda}(h_\lambda(\theta))=\mathcal L(\lambda)$; hence, an estimate of $\beta_0$ provides an estimate of the lower bound - an important quantity in the VB literature.
Let $T$ be the design matrix whose row vectors are $(1,T(\theta_i)^\top)$  and $h$ be the response vector of $h_\lambda(\theta_i)$. 
The minimum-norm ordinary least squares estimator of 
$(\beta_0,{\beta}^\top)^\top$ is 
\begin{align}\label{eq: nat grat est}
    g = T^+ h
\end{align}
where $T^+$ denotes the pseudo-inverse of $T$ \cite{Kobak2020JMLR}.
When $M>N$, $g$ is the usual OLS estimator, $g=(T^\top T)^{-1}T^\top h$ and is an {\it unbiased} estimator of $(\beta_0,{\beta}^\top)^\top$.  When $M < N$, $g=T^\top(TT^\top)^{-1}h$ and is the minimum-norm solution to the underconstrained linear equation $Tg = h$,  
and it can be shown that $g$ is a consistent estimator (in $M$) of $(\beta_0,{\beta}^\top)^\top$.

In this paper, we focus mainly on the situation where the number of $\{\theta_i\}$ samples $M$ is much smaller than the number of parameters $N$, $M \ll N$.  This is the situation observed in the vast majority of practical applications of VB in deep learning \cite{liang2019fisher,bartlett2020benign,dar2021farewell}.

Denote by $\wh\beta$ the vector $g$ after deleting its very first component.
We use $\wh\beta$ as an estimator of the natural gradient $\nabla_\l^\text{nat}{\mathcal L}(\l)$.
 This regression-based natural gradient formulation is inspired by the method of 
Salimans and Knowles \cite{salimans2013fixed} and Malago \emph{et al.} \cite{malago2013natural}
who consider a regression-based method for estimating the natural gradient when the variational distribution is limited to the exponential family.
Our method is general and can be easily applied to situations with complicated variational distributions such as neural networks.
Hereafter, for notational simplicity, we will refer to $g$ in \eqref{eq: nat grat est} as the natural gradient estimator $\wh{\nabla_\l^\text{nat}{\mathcal L}}(\l)$, although in practice one must remove the very first component of $g$ to obtain the natural gradient estimate.

\subsection{Computational Efficiency of Regression-Based Natural Gradient}
\label{sec:efficiency_nat_grad}

The regression-based approach provides an estimate of the natural gradient that avoids the requirement to compute the Fisher matrix $I_F$ analytically.  Pseudo-inversion of the short-and-wide $M \times N$, $M\ll N$ design matrix $T$ can be efficient.  The time complexity of the popular Cholesky decomposition method is  $\tilde{O}(NM^2)$, scaling favorably with $N$, particularly if $M = O(N^\alpha)$ with $0 \leq \alpha < 1/2$, where $\alpha = 0$ denotes a polylogarithmic dependence on $N$.  The method delivers significant efficiency over the naive methods to invert the $N \times N$ Fisher matrix, with time complexity of $\tilde{O}(N^d)$, where $2<d\leq 3$ (and usually $d=3$ in practice). Note that the conjugate gradient method \cite{shewchuk1994introduction,saad2003iterative}, sometimes used as a classical benchmark for quantum linear systems algorithms \cite{clader2013preconditioned}, does not apply to the natural gradient problem.  The conjugate gradient method requires that the inverted matrix is symmetric and positive definite.  The design matrix $T$ is not symmetric and its symmetrized version $\begin{pmatrix} 0 & T \\ T^\dagger & 0 \end{pmatrix}$ is not positive definite.

Theorem \ref{the: main theorem} in Section \ref{sec:Convergence} provides a convergence guarantee for VB optimization with natural gradient (Algorithm \ref{Al:algorithm 1}) under standard conditions.

\subsection{Convergence of VB with Regression-Based Natural Gradient}\label{sec:Convergence}
This section provides a convergence analysis of the VB algorithm using the regression-based natural gradient described in the previous sections. To be consistent with the conventional notation in the optimization literature, let us denote $\tilde{\mathcal L}(\lambda):=-{\mathcal L}(\lambda)$ to be the objective function, the VB optimization problem becomes
\[\min_{\lambda}\tilde{\mathcal L}(\lambda).\]
Write the estimate $\widehat{\nabla_\lambda^\text{nat}\tilde{\mathcal L}}(\lambda_t)$ of the natural gradient as
\beq\label{eq:classical NG}
\widehat{\nabla_\lambda^\text{nat}\tilde{\mathcal L}}(\lambda_t)=\nabla_\lambda^\text{nat}\tilde{\mathcal L}(\lambda_t)+\epsilon_t
\eeq
with $\epsilon_t$ the estimate error.
In Algorithm \ref{Al:algorithm 1}, by viewing that 
the variational parameter $\lambda$ is updated first then the momentum gradient, it can be equivalently rewritten as: 

For $t=0,1,...$
\begin{itemize}
\item $\lambda_{t+1}=\lambda_t-Y_t$.
\item $Y_{t+1}=\zeta_{t+1}  Y_t+\gamma_{t+1}(\nabla_\lambda^\text{nat}\tilde{\mathcal L}(\lambda_{t+1})+\epsilon_{t+1})$.
\end{itemize}
Here, $\zeta_t=\alpha_t\omega$ and $\gamma_t=\alpha_t(1-\omega)$ with $\omega$ the momentum weight and $\alpha_t$ the learning rates.
For a symmetric and positive definite matrix $\Sigma$, we denote by $\sigma_{\max}(\Sigma)$, $\sigma_{\min}(\Sigma)$ the largest and smallest eigenvalue of $\Sigma$, and by $\|\cdot\|$ the $\ell^2$-norm.

\begin{theorem}\label{the: main theorem}
Assume that the learning rates satisfy the conditions that $1\geq \alpha_0\geq \alpha_1\geq\cdots>0$, $\sum\alpha_t=\infty$ and $\sum\alpha_t^2<\infty$. Furthermore, assume that the following regularity conditions are satisfied in a neighborhood $\mathcal N$ in the  $\lambda$-space:
\begin{itemize}
    \item[(A1)] The ordinary gradient is bounded, i.e. $\|\nabla\tilde{\mathcal L}(\lambda)\|\leq b_1<\infty$.
    \item[(A2)] $0<b_2\leq\sigma_{\min}(I_F(\lambda))\leq\sigma_{\max}(I_F(\lambda))\leq b_3<\infty$. 
    \item[(A3)] The estimate errors are such that $\E(\|\epsilon_t\|)<b_4<\infty$ for all $t$, and $\sum_t\alpha_t\E\|\epsilon_t\|<\infty$.
\end{itemize}
Let $\{\lambda_t,t=0,1,...\}$ be the iterates generated from Algorithm \ref{Al:algorithm 1}.
Then, if $\lambda_0\in \mathcal N$,
\begin{equation}\label{eq:result 1}
\min_{t=1,...,n}\E\|\nabla\tilde{\mathcal L}(\lambda_t)\|\to 0,\;\;\;n\to\infty.
\end{equation}
Furthermore, if $\tilde{\mathcal L}(\lambda)$ is $L$-strongly convex on $\mathcal N$ with the minimum point $\lambda^*$, i.e.,
\[\tilde{\mathcal L}(\lambda)\geq \tilde{\mathcal L}(\eta)+\nabla \tilde{\mathcal L}(\eta)^\top(\lambda-\eta)+\frac{L}{2}\|\lambda-\eta\|^2,\;\;\forall \lambda,\eta\in \mathcal N,\]
then 
\begin{equation}\label{eq:result 2}
\min_{t=1,...,n}\E\|\lambda_t-\lambda^*\|\to 0,\;\;\;n\to\infty.
\end{equation}
\end{theorem}
The practical implication of condition (A3) is that, as 
the natural gradient estimate $\widehat{\nabla_\lambda^\text{nat}\tilde{\mathcal L}}(\lambda_t)$ is no longer unbiased when $M<N$, it should be made more and more accurate  along with iterations $t$.
That is, one might need to gradually  increase 
the sample size $M$ in the estimate \eqref{eq: nat grat est}. 
The result in \eqref{eq:result 1} says that the algorithm does stop at a stationary point of the target $\tilde{\mathcal L}(\lambda)$. 
The result in \eqref{eq:result 2} says that, under the strong convexity condition, the algorithm converges to the minimum point $\lambda^*$.
The proof can be found in the Appendix.

Amari \cite{amari1998natural} demonstrated that (non-stochastic) natural gradient is Fisher-efficient and, since Amari's work was published, natural gradient has been shown to require far fewer iterations than other methods for many models (see, e.g., Ref.~\cite{martens2014new} and references therein). However, a rigorous treatment of the convergence speed of (stochastic or non-stochastic) natural gradient is still an active open area of research. %The frameworks available for stochastic (Euclidean) gradient descent, such as the ones used by Sweke \emph{et al.}~\cite{sweke2020stochastic} in the context of quantum-classical optimization, do not extend naturally to stochastic natural gradient descent.   

\section{Quantum Speedup of Variational Bayes}
\label{sec:QVB}

One of the major benefits of the linear regression formulation of natural gradient estimation is that it enables quantum speedup of VB.  Multiple quantum algorithms for efficient estimation of linear regression exist \citep{wiebe2014quantum,schuld2016prediction,wang2017quantum,liu2019quantum,yu2019improved,potok2021adiabatic}, under certain conditions delivering material speedup over the best classical solutions.   The quantum linear regression algorithms are based on quantum linear systems algorithms \cite{harrow2009quantum,ambainis2012variable,clader2013preconditioned,kerenidis2016quantum,wossnig2018quantum,childs2017quantum,subacsi2019quantum,gilyen2019quantum}, the most efficient of which, when provided data input in the form of quantum states, deliver exponential speedup over the best classical algorithms.   Accounting for quantum state preparation and readout, quantum computers can estimate quantum natural gradient polynomially faster than the best classical counterparts, even though the regression formulation supports efficient classical methods.  In this section, we demonstrate that, for certain commonly met conditions, where the sample size $M$ and the condition number $\kappa$ of design matrix  $T$ (the ratio of the highest and lowest singular values of $T$, $\kappa \equiv \tau_{max}/\tau_{min}$) scale as a fractional power of the model size $N$, quantum algorithms (known at the time of writing) can deliver further polynomial computational advantage.

\subsection{Quantum Natural Gradient Estimation}
\label{sec:qng estimation}

The expression for the natural gradient estimate $g = T^+h$ in (\ref{eq: nat grat est}) is a solution to the (underconstrained) linear system $Tg = h$.  Quantum linear systems algorithms, first proposed by Harrow, Hassidim, and Lloyd (HHL) \citep{harrow2009quantum}, solve the quantum analogue of the linear systems problem.  The algorithms take as an input the quantum state $\ket{h}$, which encodes the vector $h$ in amplitude encoding \cite{schuld2018supervised,lopatnikova.tran:2022:introduction}, and output the quantum state $\ket{g} \equiv \ket{T^+ h}$, which encodes the natural gradient estimate $g$ in a quantum state in amplitude encoding.  

The runtime and resource requirements of quantum linear systems algorithms depend on the dimensions $M$ and $N$ of the matrix $T$, its condition number $\kappa$ and the required precision $\epsilon$.  The runtime of HHL quantum linear systems algorithm with updated Hamiltonian simulation \cite{berry2015hamiltonian} scales as $O(M\kappa^2/\epsilon\;\text{poly}\,\log (N\kappa/\epsilon))$ \citep{childs2017quantum}.  When $M = O(\text{poly}\;\log(N))$ and $\kappa = O(\text{poly}\;\log(N))$, the HHL algorithm outputs $\ket{g}$ exponentially faster than the best classical algorithm outputs $g$.  An alternative quantum linear systems algorithm by Wossnig, Zhao, and Prakash (WZP) \cite{wossnig2018quantum} has time complexity of $O(\kappa^2 \sqrt{N} \text{poly}\;\log N/\epsilon)$ for a bounded spectral norm of $T$; this method is attractive relative to HHL when $M > O(\sqrt{N})$.  The algorithms assume sparse access to a black-box Hamiltonian oracle and an address oracle \citep{aharonov2003adiabatic} to make the entries of the design matrix $T$ available to the algorithm.  Note that the need to estimate the non-zero entries of $T$ limits the efficiency of any algorithm -- classical or quantum.  The computation of $T$ requires $O(NM)$ operations and, unless an efficient way of estimating $T$ is found (e.g.~for specific variational distributions $q_\lambda(\theta)$), $O(NM)$ bounds from below the computational complexity of natural gradient estimation by classical, quantum, or the so-called quantum-inspired methods.

When $\kappa = O(\text{poly}\,\log(N))$ and $M = O(\text{poly}\,\log(N))$ and the required precision does not have to be high, i.e. $\epsilon^{-1} =O(\text{poly}\,\log(N))$, quantum-inspired algorithms~\cite{gilyen2018quantum} can be most efficient.  Quantum-inspired algorithms are randomized classical algorithms designed to emulate some computational advantages of quantum algorithms using sampling that mimics quantum measurements \cite{frieze2004fast}.  The method's time complexity does not (directly) depend on the size of the problem $N$, but scales as $\tilde{O}(\kappa^{16} k^6\|T\|_F^6/\epsilon^6)$, where $\tilde{O}(f(x)) \equiv \tilde{O}(f(x)\text{poly}\;\log(g(x)))$, $k$ is the rank of $T$ (we can assume $k=M)$, and $\|\cdot\|_F$ denotes the Frobenius norm.

\subsection{Hybrid Quantum-Classical VB Algorithm}
\label{sec:hybrid QVB alg}

The quantum linear system algorithm provides an efficient way to estimate the natural gradient $g$, but its output is a quantum state $\ket{g}$.  To use $\ket{g}$ in VB, we face a design choice: either use $\ket{g}$ in a quantum algorithm -- i.e.~run the entire VB gradient descent algorithm on a quantum computer -- or read out the information encoded in $\ket{g}$ for use on a classical computer.  Even though readout can be a resource-intensive step, we propose to use a hybrid quantum-classical approach, illustrated in Figure~\ref{fig:algo}, where the quantum computer calculates $\ket{g}$, and the classical computer uses the information to update parameters $\lambda$, to sample from the variational distribution $q_\lambda(\theta)$, and to estimate the likelihood vector $h$ and the design matrix $T$.

\begin{figure*}[ht]
\centering
\includegraphics[scale = 0.65,trim={0 1cm 2cm 0},clip]{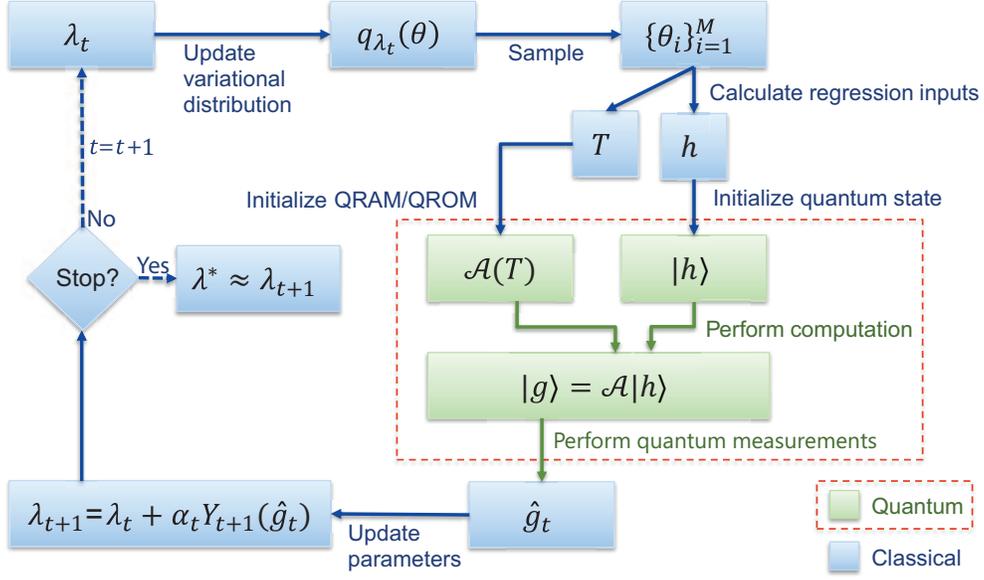}
\caption{A schematic representation of a single iteration of the hybrid quantum-classical VB algorithm.  The update of the variational parameter vector $\lambda_t$ starts in the upper left corner and proceeds clockwise.  The dashed line highlights the steps performed by a quantum computer, where $\mathcal{A}$ stands for the quantum algorithm.  The output of the quantum algorithm is an consistent estimate of the natural gradient $\hat{g}_t$.  This estimate feeds into the momentum gradient $Y_{t+1}$, defined in Eq.~(\ref{eq:mom grad}), used to update the variational parameter vector and obtain its next iteration $\lambda_{t+1}$.  If the stopping criterion is reached, the algorithm outputs $\lambda_{t+1}$ as the desired estimate of the optimal variational parameter vector; if the stopping criterion is not satisfied, the algorithm proceeds with the update of $\lambda_{t+1}$. } 
\label{fig:algo}
\end{figure*}

The hybrid quantum-classical approach is more suitable for VB gradient descent because iterative quantum algorithms are, in general, resource intensive.  Quantum transformations are linear; to create effectively non-linear transformations repeated quantum registers are used, some of which are destroyed using projective quantum measurements \citep[see, e.g.][]{gilyen2019quantum,rebentrost2019quantum}.  If the projective measurement yields the desired result with probability $p<1$ (usually $p\leq 1/2)$, the measurement destroys, on average, $1/p-1$ quantum registers.   For an iterative quantum algorithm to yield the desired result after $T$ iterations, $O(p^{-T})$ qubit registers are required.  Therefore, end-to-end quantum algorithms work for iterative computations only if a few iterative steps are sufficient. Bausch \cite{bausch2020recurrent} proposes a recurrent method with post-selection in the context of quantum Recurrent Neural Networks and argues that, for approximate post-selection it is possible to reduce this overhead.  Another possible approach is to run the iterative algorithm without post-selection for $T$ steps and then amplify the desired end state using Quantum Amplitude Amplification (QAA) \citep{brassard2002quantum}.  Even with these improvements, the exponential dependence on the number of iterations $T/z$ (where $z=2$ for QAA method) remains. For example, Rebentrost et al. \cite{rebentrost2019quantum} propose an end-to-end quantum algorithm to find the minimum of a homogeneous polynomial using Newton’s method (which leverages the density matrix exponentiation method of Lloyd \emph{et al.}  \cite{lloyd2014quantum}), suitable for situations where only a few iterative steps are required.  If an optimization requires more than a few iterative steps, a hybrid quantum-classical algorithm, such as the one proposed in this paper, is a better choice, provided it can handle the classical-quantum-classical handoff efficiently.

The hybrid quantum-classical VB algorithm with quantum natural descent is summarized in Algorithm~\ref{Al:hybrid qc algo}.   It extends the classical Algorithm~\ref{Al:algorithm 1} to embed a quantum subroutine to estimate the natural gradient.  We consider two versions of the quantum natural gradient estimator:  The first, Algorithm \ref{Al:qng algo}, includes full readout of the output state $\ket{g}$.  The second, Algorithm \ref{Al:qng-gs algo}, includes partial readout of $\ket{g}$ focused on the most significant coordinates of $g$.  We describe the subroutines in the next two subsections.

\subsection{Quantum Natural Gradient with Full Readout}
\label{sec:full_readout}

The quantum natural gradient subroutine with full readout (Algorithm \ref{Al:qng algo}) starts with the initialization of the quantum state $\ket{h}$ and making available the entries of the design matrix $T$.   Note that because $h$ is a vector of length $M$ and $M\ll N$, preparation of the input state $\ket{h}$ requires only $O(M/\log M)$ operations, negligible relative to the complexity of other parts of the algorithm.  As discussed above, the computation of the $M \times N$ entries of matrix $T$ bounds from below the time complexity of both classical and quantum algorithms for natural gradient estimation.  Once matrix $T$ is stored on a classical computer, time-efficient methods have been proposed to make it available for quantum computation either in a quantum random access memory (QRAM) \cite{giovannetti2008quantum,giovannetti2008architectures,hann2021resilience} or a classical read-only memory for quantum access (QROM) \cite{kerenidis2016quantum}.    With these inputs, the quantum computer executes a quantum linear systems algorithm (e.g.~HHL or WPZ) to perform singular value inversion and the computation of the quantum state $\ket{g}$.  The last step of the quantum subroutine is the readout of state $\ket{g}$ using a series of quantum measurements. 
\begin{algorithm}[H]
\caption{VB with Quantum Natural Gradient\label{Al:hybrid qc algo}}
\vspace{-1.5ex}
\begin{enumerate}[leftmargin=*]
\itemsep0em 
\item On a classical computer, initialize $\l_0$ and momentum gradient $Y_0$. 
\item For $t=0,1,...$and until the stopping criterion is met:
\vspace{-1.5ex}
  \begin{enumerate}[a.,leftmargin=*]
  \itemsep0em 
	  \item Generate $\theta_j\sim q_{\lambda^{(t)}}(\theta)$, $j=1,...,M$, and form the matrix $T$ and vector $h$.
	  \item Use a classical, quantum (Algorithm \ref{Al:qng algo} or \ref{Al:qng-gs algo}), or quantum-inspired algorithm to obtain the natural gradient estimate $\widehat g_t$.
	  \item Update the momentum gradient
	  \[Y_{t+1}=\omega Y_t+(1-\omega)\wh{g_t}.\label{eq:momgrad}\]
	  \item Update the variational parameter
	  \bea%\label{eq:update rule}
	  \lambda_{t+1}=\lambda_{t}+\alpha_t Y_{t+1}.
	  \eea	  
  \end{enumerate}
\end{enumerate}
\end{algorithm}

Since quantum measurement of state $\ket{g}$ in the basis $\{\ket{i}\}_{i=1}^N$ yields the state $\ket{i}$ with probability $|g_i|^2/\|g\|^2$, repeated measurements allow us to estimate $|g_i|/\|g\|$ -- the normalized absolute values of each gradient coordinate.  To extract the sign of $g_i$, an additional step is required.  
At the beginning of the quantum computation, we pull in an additional auxiliary qubit and apply a Hadamard gate \cite[see, e.g.][]{nielsen2002quantum,lopatnikova.tran:2022:introduction} to create two computational branches.  In the computational branch of the auxiliary qubit in state $\ket{0}$ we run the HHL or WPZ algorithm to create the state $\ket{g}$; in the $\ket{1}$-branch we create the uniform state $\ket{\gamma}\equiv \frac{1}{\sqrt{N}}\sum_{i=1}^{N}\ket{i}$.  The resulting state is
\begin{align}
    \frac{1}{\sqrt{2}}\ket{0}\ket{g}+\frac{1}{\sqrt{2}} \ket{1}\ket{\gamma}.
\end{align}
which is equivalent to
\begin{eqnarray*}
\frac{1}{2}\sum_{i=1}^N(g_i+\frac{1}{\sqrt{N}})\ket{+}\ket{i}+\frac{1}{2}\sum_{i=1}^N(g_i-\frac{1}{\sqrt{N}})\ket{-}\ket{i},
\end{eqnarray*}
where $\ket{\pm} \equiv \frac{1}{\sqrt{2}}(\ket{0}\pm\ket{1})$ represent Hadamard basis states \cite{nielsen2002quantum,lopatnikova.tran:2022:introduction} of the auxiliary qubit.  To streamline notation and without loss of generality, we assumed $\|g\|=1$ (we discuss the implications of measuring $g/\|g\|$ rather than $g$ at the end of this subsection).

\begin{algorithm}[H]
\caption{Quantum Natural Gradient with Full Readout\label{Al:qng algo}}
\vspace{-1.5ex}
\begin{enumerate}[leftmargin=*]
\itemsep0em 
\item Initialize registers and inputs:
\vspace{-1.5ex}
\begin{enumerate}[a.,leftmargin=*]
\itemsep0em 
\item On a classical computer, initialize two counter vectors -- length $N$ integer vectors $n^+$ and $n^-$ to $\vec{0}$.  
\item On a quantum computer, initialize a data register of $\text{ceil}(\log N)$ qubits and an auxiliary qubit.  
\item On a quantum or classical computer, initialize appropriate data structure for efficient quantum access to elements of $T$
\end{enumerate}
\item Obtain natural gradient estimate $\widehat g$:  
\vspace{-1.5ex}
\begin{enumerate}[a.,leftmargin=*]
\itemsep0em 
\item For $j=0,1,\ldots$ and until stopping criterion is met 
\begin{itemize}
\item Put the auxiliary qubit in the state $\frac{\ket{0}+\ket{1}}{\sqrt{2}}$.
\item In the data register, in the $\ket{1}$-branch of the auxiliary qubit, create a uniform state $\ket{\gamma}$.
\item In the data register, in the $\ket{0}$-branch of the auxiliary qubit:
\begin{itemize}
\item Initialize the state $\ket{h}$. 
\item Perform quantum matrix inversion using HHL or WZP algorithm. (Additional auxiliary registers and qubits are appended as needed in this step.) Obtain quantum state $\ket{g}$.
\end{itemize}
\item Perform simultaneous quantum measurement of the data register in the computational basis and the auxiliary qubit in the Hadamard basis. Receive paired outcome $[i;p]$, were $i = 0,\ldots,N-1$ and $p \in \{+,-\}$.
\item Update counter vector $n^p_i = n^p_i +1$.
\end{itemize}
\item Construct the unbiased estimate  $\widehat g$ as in \eqref{eq: nat_grad_est}.
\end{enumerate}
\end{enumerate}
\end{algorithm}

Simultaneous measurements of the register holding the superposition of $\ket{g}$ and $\ket{\gamma}$, which we call the data register, and the auxiliary register yield results that comprise an unbiased estimator of ${g}_i$.   We measure the data register in the computational basis and the auxiliary qubit in the Hadamard basis.  
The probability of measuring the $\ket{\pm}$ state
in the auxiliary qubit and $\ket{i}$ in the data register is $|g_i\pm\frac{1}{\sqrt{N}}|^2/4$. 
We perform repeated measurements and count the outcomes.
Let $n^\pm_i$ be the number of measurements of the data register that yield the state $\ket{i}$ when the measurement of the auxiliary qubit yields $\ket{\pm}$, and $n_T = \sum_{i=1}^N(n^+_i+n^-_i)$ be the total number of measurements taken.  Let $\mathbb{E}_m$ indicate expectation over measurement results.  Then,
\begin{align}
&\mathbb{E}_m (n^+_i-n^-_i) =  \nonumber \\ 
& = n_T\Big(\frac{|g_i+1/\sqrt{N}|^2}{4}-\frac{|g_i-1/\sqrt{N}|^2}{4}\Big)=n_T\frac{g_i}{\sqrt{N}},\label{eq:nat_grad_est_gs}
\end{align}
which implies that measurement results yield an unbiased estimator of $g_i$:
\begin{align}\label{eq: nat_grad_est}
    \widehat{g}_{i} = \sqrt{N}\frac{n^+_i-n^-_i}{n_T},\;\;\; \mathbb{E}_m (\widehat{g}_{i}) & = g_i. 
\end{align}
We denote by $\wh g=(\wh{g}_{1},...,\wh{g}_{N})^\top$ the quantum natural gradient which is a classical vector. 
%The following lemma summarizes the two important properties of 
%the quantum natural gradient $\wh g$:
%\begin{lemma}\label{lem:QNG properties}
%\begin{itemize}
%    \item[(i)] The quantum natural gradient is unbiased, i.e., $\E_m(\widehat{g})=\nabla_\l^\text{nat}{\mathcal L}$.
%    \item[(ii)] If the classical natural gradient is bounded, i.e. $\V(g_i)<\infty$ for all $i$, then $\V(\widehat g_{i})<\infty$ for all $i$.
%\end{itemize}    
%\end{lemma}

Technically, this algorithm results in a unbiased estimator of the normalized natural gradient, i.e. $g/\|g\|$, because of the unit norm requirement of quantum states in quantum computation. This makes quantum natural gradient suit naturally to gradient clipping (Section \ref{sec:vb_with_ng}) - a method used in stochastic gradient descent that clips off the length of the gradient vector to help stabilize the optimization \cite{Goodfellow-et-al-2016}.

\iffalse
\subsection{Convergence of VB with Quantum Natural Gradient}
\label{sec:Convergence for QNG}
This section provides a convergence analysis of the hybrid quantum-classical VB method in Algorithm \ref{Al:hybrid qc algo}.
We note that, the quantum natural gradient $\widehat{g}_t$, conditional on the $\theta$ samples, is an unbiased estimator of the classical regression-based natural gradient $g_t$. Hence, we can write $\wh g_t$ as
\[\widehat{g}_t=\nabla_\lambda^\text{nat}\tilde{\mathcal L}(\lambda_t)+\epsilon_t+u_t,\]
where $\epsilon_t$ is the estimate error as in \eqref{eq:classical NG}, and $u_t$ satisfies $\E(u_t|\mathcal F_{t})=0$ with $\mathcal F_{t}$ the $\sigma$-field generated by $\{\lambda_s,s\leq t\}$.

\begin{theorem}\label{the: main theorem QNG}
Assume that the assumptions in Theorem \ref{the: main theorem} are satisfied. Furthermore, assume that 
\[\E\|u_t\|<\infty.\]
Then the conclusions in Theorem \ref{the: main theorem} hold.
\end{theorem}
\fi

\subsection{Quantum Natural Gradient with Gauss-Southwell Rule}
\label{sec:gauss_southwell}

Full readout of $g/\|g\|$ requires $O(N/\epsilon^2)$ quantum measurements to estimate the vector $g$ within precision $\epsilon$.  Full readout, however, may not be required when estimation of top gradient coordinates (by absolute value) is sufficient for VB convergence. We propose a new method we call the quantum Gauss-Southwell rule developed specifically for quantum natural gradient estimation (Algorithm~\ref{Al:qng-gs algo}).  The method extracts one of the highest-value directions of the natural gradient estimate vector $g$ within precision $\epsilon$ using $O(\sqrt{N}/\epsilon)$ measurements.  

The quantum Gauss-Southwell rule is a quantum analogue of the Gauss-Southwell rule for efficient classical coordinate descent \cite{nutini2015coordinate}. At each step of coordinate descent, the classical Gauss-Southwell rule selects the top gradient coordinate (by absolute value).  Classical Gauss-Southwell rule can deliver efficient linear convergence even for high-dimensional models; its primary shortcoming is its computational intensity relative to randomized coordinate descent, because it requires the computation of the full gradient vector.  The quantum approach estimates the full gradient vector highly efficiently in quantum parallel, making quantum Gauss-Southwell rule efficient and able to handle \emph{natural gradient} Gauss-Southwell descent, infeasible on a classical computer.

\begin{algorithm}[H]
\caption{Quantum Natural Gradient with Gauss-Southwell Rule\label{Al:qng-gs algo}}
\vspace{-1.5ex}
\begin{enumerate}[leftmargin=*]
\itemsep0em 
\item Initialize registers and inputs:
\vspace{-1.5ex}
\begin{enumerate}[a.,leftmargin=*]
\itemsep0em 
\item On a quantum computer, initialize a data register of $\text{ceil}(\log N)$ qubits and an auxiliary qubit.  
\item On a quantum or classical computer, initialize appropriate data structure for efficient quantum access to elements of $T$
\item In the data register, initialize the state $\ket{h}$.
\end{enumerate}
\item Determine high-value gradient coordinate $i$:
\vspace{-1.5ex}
\begin{enumerate}[a.,leftmargin=*]
\itemsep0em 
\item In the data register, perform quantum matrix inversion using HHL or WZP algorithm. (Additional auxiliary registers and qubits are appended as needed in this step.)
\item Perform  quantum measurement of the data register in the computational basis. Receive paired outcome $i$, were $i = 0,\ldots,N-1$.
\end{enumerate}
\item Obtain natural gradient coordinate estimate $\widehat{g_i}$:
\vspace{-1.5ex}
\begin{enumerate}[a.,leftmargin=*]
\itemsep0em 
\item \label{step:start}Uncompute the data register
\item Put the auxiliary qubit in the state $\frac{\ket{0}+\ket{1}}{\sqrt{2}}$.
\item In the data register, in the $\ket{1}$-branch of the auxiliary qubit, create a uniform state $\ket{\gamma}$.
\item In the data register, in the $\ket{0}$-branch of the auxiliary qubit:\label{step:end}
\begin{itemize}
\item Initialize the state $\ket{h}$. 
\item Perform quantum matrix inversion using HHL or WZP algorithm. (Additional auxiliary registers and qubits are appended as needed in this step.)  Obtain quantum state $\ket{g}$.
\end{itemize}
\item Perform Quantum Amplitude Estimation with $\ket{i}\ket{+}$ as the reference state.  Receive estimate of $\frac{1}{2}\big|\frac{g_i}{\|g\|}+\frac{1}{\sqrt{N}}\big|$.
\item Repeat Steps \ref{step:start} through \ref{step:end}  
\item Perform Quantum Amplitude Estimation with $\ket{i}\ket{-}$ as the reference state.  Receive estimate of $\frac{1}{2}\big|\frac{g_i}{\|g\|}-\frac{1}{\sqrt{N}}\big|$.
\item Construct the unbiased estimate $\widehat{g_i}$ as in \eqref{eq:gs_gi}.
\end{enumerate}
\end{enumerate}
\end{algorithm}

Let $\ket{g}$ be the output state of the quantum natural gradient estimation algorithm.  We measure $\ket{g}$ in the computational basis.  The measurement yields a computational basis state $\ket{i}$ with probability $|g_{i}|^2$.  The quantum measurement naturally yields the index of one of the highest absolute-value gradient directions with high probability, analogously to the classical Gauss-Southwell coordinate descent which selects the highest absolute-value gradient direction for the parameter update.  Having determined the index $i$ of a significant gradient coordinate, we use Quantum Amplitude Estimation \cite{brassard2002quantum} with target state $\ket{i}$ to extract the value of $g_{i}$ in $O(\sqrt{N}/\epsilon)$ operations.   Applying Quantum Amplitude Estimation directly to the state $\ket{g}$ yields the absolute value $|g_i|$ rather than the signed value $g_i$.  As in the case of full readout (Section~\ref{sec:qng estimation}), in order to extract the sign of $g_i$, we append an auxiliary qubit in order to be able to estimate  $\frac{1}{2}\big|\frac{g_i}{\|g\|}\pm\frac{1}{\sqrt{N}}\big|$ as described in Algorithm~\ref{Al:qng-gs algo}:
\begin{align}
    \frac{1}{4}\Big|\frac{g_i}{\|g\|}+\frac{1}{\sqrt{N}}\Big|^2-\frac{1}{4}\Big|\frac{g_i}{\|g\|}-\frac{1}{\sqrt{N}}\Big|^2 = \frac{1}{\sqrt{N}}\frac{g_i}{\|g\|} \label{eq:gs_gi}
\end{align}

Formal proof of convergence of quantum VB with Gauss-Southwell rule is outside the scope of this work and is a topic of future research.

\subsection{Complexity of Quantum Natural Gradient Estimation Algorithms}

The hybrid quantum-classical VB algorithm provides a flexible framework that can leverage one of a number of quantum, classical, or even classical quantum-inspired \cite{tang2018quantum} algorithms for natural gradient estimation.  The key factors in determining which natural gradient estimation algorithm is to be used are the number of samples $M$ and the condition number $\kappa$ of the design matrix $T$.   The desired precision $\epsilon$ is also a driver of computational complexity, but it is not the limiting variable in the the stochastic approximation to the natural gradient.  
We compare two quantum algorithms -- HHL and WPZ, with two measurement schemes (full readout and Gauss-Southwell rule) to read off $\hat{g}$ -- with a classical singular value decomposition (SVD) algorithm and a quantum-inspired algorithm \cite{gilyen2018quantum}.    

Table~\ref{tab:complexity comparison} summarizes the time complexity of algorithms for natural gradient estimation. The time complexity of quantum algorithms takes into account quantum state preparation and readout.  Each of the algorithms we consider is the most efficient algorithm in some regime defined by the scaling of $M$ and $\kappa$ with the size of the model $N$.   Quantum algorithms are most efficient when $M = \tilde{O}(N^\alpha)$, where $0< \alpha <1$, and $\kappa < C_\kappa M$, where $C_\kappa$ is a constant.  When $\alpha=0$ and $\kappa =\tilde{O}(\text{poly}\,\log(N))$, i.e.~the matrix $T$ is sparse and well-conditioned, quantum-inspired algorithms are most efficient.  Table~\ref{tab:efficient algos} summarizes the most efficient algorithms given $M$ and $\kappa$.  The summary assumes that the scaling of all algorithms with problem size is bounded from below by the need to estimate the design matrix $T$; should more efficient ways to estimate non-zero entries of $T$ become available, quantum algorithms will become the most efficient choice across most combinations of $M$ and $\kappa$.

\begin{table}[H]
%\small 
\caption{Time complexity of pseudoinverse algorithms as a function of parameter vector size $N$, number of samples $M$, condition number $\kappa$ of the design matrix $T$ (the ratio of the highest and lowest singular values of $T$, $\kappa \equiv \tau_{max}/\tau_{min}$), the required precision $\epsilon$, and the Frobenius norm $\|T\|_F$.  Quantum algorithms by Harrow, Hassidim, and Lloyd \cite{harrow2009quantum,berry2015hamiltonian} (\emph{HHL}) and Wossnig, Zhao, and Prakash \cite{wossnig2018quantum} (\emph{WZP}) with full readout (\emph{FR}) or Gauss-Southwell rule (\emph{GR}) are compared with classical SVD algorithms (\emph{CL}) and quantum-inspired classical algorithm by Gilyen, Lloyd, and Tang \cite{gilyen2018quantum} (\emph{QI}).  \label{tab:complexity comparison}} 
\centering
 \begin{tabular}{ l l  l } 
 \hline\\[-2.0ex]
\multicolumn{2}{l}{Algorithm} & \multicolumn{1}{c}{Complexity} \\[0.5ex] 
 \hline\hline\\[-2.0ex]
 \multicolumn{3}{l}{\emph{Quantum:}}\\[0.5ex]
%\hline\\[-2.0ex]
& HHL-FR & $\max(O(\frac{NM\kappa^2}{\epsilon^3} \text{poly}\log(\frac{N\kappa}{\epsilon})),O(NM))$  \\[0.5ex]
& HHL-GS & $\max(O(\frac{\sqrt{N}M\kappa^2}{\epsilon^2} \text{poly}\log(\frac{N\kappa}{\epsilon})),O(NM))$ \\[0.5ex]
& WZP-FR & $\max(O(\frac{N^{3/2}\kappa^2}{\epsilon^3} \text{poly}\log(N)),O(NM))$ \\[0.5ex]
& WZP-GS & $\max(O(\frac{N\kappa^2}{\epsilon^2} \text{poly}\log(N)),O(NM))$ \\[1.5ex]
%\hline\\[-2.0ex]
 \multicolumn{3}{l}{\emph{Classical:}}\\[0.5ex]
\phantom{a}& CL & $\tilde{O}(NM^2)$ \\ [0.5ex]
%\hline\\[-2.0ex]
\multicolumn{3}{l}{\emph{Quantum-Inspired:}}\\[0.5ex]
%\hline\\[-2.0ex]
& QI & $\tilde{O}(M^6 \kappa^{16} \|T\|_F^6/\epsilon^6)$ \\[0.5ex]
\hline
 \end{tabular}
\end{table}

\begin{table*}[t]
%\small 
\caption{Relative efficiency of pseudoinverse algorithms across design matrix sparsity/rank and conditioning regimes. Table lists most efficient algorithm(s) for each regime.  The number of samples $M$ determines the sparsity or rank of the design matrix $T$.  The conditioning number $\kappa$ is the ratio of the maximum and minimum singular values of $T$, $\kappa = \tau_{max}/\tau_{min}$.  Scaling of $M$ and $\kappa$ with model size $N$, $M=\tilde{O}(N^\alpha)$ and $\kappa=\tilde{O}(N^\delta)$, defines regimes that favor one or a few of selected algorithms.  The scaling exponents $\alpha$ and $\delta$ equal to $0$ denote polylogarithmic dependence on $N$.  The algorithms include: Classical SVD algorithms (\emph{CL}); Quantum-inspired -- Gilyen, Lloyd, and Tang \cite{gilyen2018quantum} (\emph{QI}); Quantum -- Harrow, Hassidim, and Lloyd \cite{harrow2009quantum,berry2015hamiltonian} (\emph{HHL}) and Wossnig, Zhao, and Prakash \cite{wossnig2018quantum} (\emph{WZP}) with full readout (\emph{FR}) or Gauss-Southwell rule (\emph{GR}).  (Quantum algorithm with no readout scheme specified indicates that the algorithm is favored independently of readout method.)\label{tab:efficient algos}} 
\centering
 \begin{tabular}{lll  l l l ll} 
\hline\\[-2ex]
 & & \multirow{2}{*}{\phantom{ab}} & \multicolumn{4}{c}{$M = \tilde{O}(N^\alpha)$} & \\[0.5ex] 
\cline{4-7}\\[-2ex]
& &  & $\alpha = 0$\phantom{abcdefg} & $0<\alpha\leq 1/4$\phantom{abcd} & $1/4<\alpha\leq 1/2 $\phantom{ab} & $1/2<\alpha<1 $\phantom{ab} & \\[0.5ex] 
\hline
\hline\\[-2ex]
& \multirow{6}{*}{$\kappa=\tilde{O}(N^\delta)$\phantom{a}} & $\delta = 0$ & QI & HHL, WZP-GS \phantom{a}& HHL, WZP \phantom{a} & HHL, WZP & \\[0.5ex] 
  \cline{3-7}\\[-2ex]
& & $0<\delta\leq 1/4$ & CL, HHL-GS \phantom{a}& HHL-GS & HHL-GS, WZP-GS & HHL-GS, WZP & \\[0.5ex] 
  \cline{3-7}\\[-2ex]
& & $1/4<\delta\leq 1/2$\phantom{a} & CL & CL, HHL-GS & CL, HHL-GS  & WZP-GS & \\
& & & & & WZP-GS \phantom{a} & & \\[0.5ex]
  \cline{3-7}\\[-2ex]
& & $1/2<\delta<1 $ & CL & CL & CL & CL&\\
& & & & & & WZP-GS &\\[0.5ex]
\hline
 \end{tabular}
\end{table*}

\section{A Numerical Example: Stiefel Neural Network VB}\label{sec: numerical examples}

This section provides a numerical example that highlights the need for quantum speed up to unlock further applications in deep learning.  We demonstrate, using classical simulation, how quantum natural gradient can enable inference on a high-dimensional dataset using an expressive, but computationally demanding deep neural network VB framework.

The recent VB literature has called for the use of flexible and expressive variational distributions $q_\lambda$ that are able to approximate sufficiently well a wide range of posterior distributions. We consider such a flexible variational distribution by the following construction
\beq\label{eq:Stiefel NL}
Z_0\sim N(0,I),\;\;Z_k=\zeta_k(W_kZ_{k-1}+b_k),\;\;k=0,...,K,
\eeq
where $N(0,I)$ denotes the multivariate standard normal distribution, the $\zeta_k$ are activation functions such as tanh or sigmoid, $W_k$ and $b_k$ are coefficient matrices and vectors respectively. The distribution of $Z_K$ gives us the variational distribution $q_\lambda$ with $\lambda=(W_1,b_1,...,W_K,b_K)$. 
The construction in \eqref{eq:Stiefel NL} is a neural network with $K$ layers and can be considered as an example of normalizing flows -
a class of methods for constructing expressive probability distributions via a composition of simple bijective transformations \citep{Papamakarios:2021}.
There are two desirable properties of a normalizing flow: the transformation from $Z_{k-1}$ to $Z_k$ must be invertible and, for computational efficiency, it must be easy to compute the Jacobian determinant. To this end, we impose the following constraint on the $W_k$:
\beqn
W_k^\top W_k = I, k =1,...,K,
\eeqn
that is, $W_k$ belongs to the Stiefel manifold.
This constraint makes it easy to compute the inverse $Z_{k-1}=W_k^\top(\zeta_k^{-1}(Z_k)-b_k)$, and the Jacobian is a diagonal matrix.
We refer to the construction of probability distributions in \eqref{eq:Stiefel NL} as the Stiefel normalizing flow or the Stiefel neural network.   
As $\lambda$ has a rich geometric structure, adaptive learning methods such as Adam and AdaGrad fail to work, leaving the natural gradient method the only option for training $\lambda$.

Being a deep neural network, the Stiefel normalizing flow is expressive and can approximate a wide range of probability distributions. However, it is extremely computationally expensive to use 
this normalizing flow in VB as the number of parameters in $\lambda$ is quadratic in the number of model parameters, which poses a real challenge for the current computational technologies.

This section simulates and demonstrates the performance of quantum natural gradient for the Stiefel neural network VB in a small dimensional setting.
We make it clear upfront that this numerical example is not run on a real quantum computer; despite many hardware advances made recently, the availability of such a computer for general use is still years away. This example uses a quantum-like surrogate to test the performance of the quantum measurement step for reading off the quantum state $\ket{g}$ into a classical vector $\widehat{g}$ as in \eqref{eq: nat_grad_est}.
Given the manifold constraints of the $W_k$, we use the VB method on manifolds of Tran \emph{et al.} \cite{Tran:2019VB_manifold}. 

We run the simulated algorithm on a standalone basis instead of contrasting its efficiency against that of alternative methodologies, such as stochastic gradient descent or adaptive learning methods (e.g. Adam or AdaGrad), because these methodologies fail to work for the Stiefel neural network.  

We use the gene expression dataset, Colon,
from \url{http://www.csie.ntu.edu.tw/~cjlin/libsvmtools/datasets/binary.html}. 
This dataset has 62 observations on a binary response variable and 2000 covariates (we only use the first 100 covariates, plus an intercept, in this example).
We consider fitting a logistic regression model to this dataset, and approximate its posterior by a Stiefel neural network with $K=2$, hence, the size of $\lambda$ is $N=20,604$.
We deliberately select $K=2$ to keep 
the size of the problem numerically manageable on a classical computer;
and focus on simulating the less-explored quantum measurement step, which is the computational bottleneck of the HHL algorithm.
More precisely, we compute $ g$ classically as in \eqref{eq: nat grat est}, normalize it, then estimate the $\widehat g_i$ as in \eqref{eq: nat_grad_est}. That is, the frequency numbers $n_i^+$ and $n_i^-$ are obtained from a multinomial distribution with probabilities $|g_i\pm\frac{1}{\sqrt{N}}|^2/4$. The number of measurements used in this example is $n_T=500$, with $M=1000$ Monte Carlo samples.  Note that the present example does not really satisfy the situation $M \ll N$ because we keep $N$ relatively small due to the complexity of the Stiefel neural network and the limitations of classical computers.  The condition $M \ll N$ applies to large-scale deep learning applications where the number of parameters $N$ is hundreds of thousands or higher and $M \approx 1000$ remains sufficient.

\begin{figure}[ht]
\centering
\includegraphics[width=.5\textwidth,height=.27\textheight]{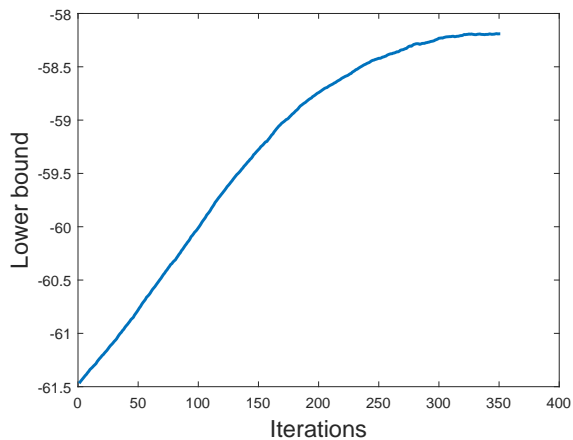}
\caption{Lower bound over iterations for Stiefel neural network VB with a classical simulation of quantum natural gradient. The smoothly-increasing lower bound indicates a well-behaved training VB procedure.   The underlying dataset has 62 observations on a binary response variable and 2000 covariates.   The present example uses the first 100 covariates, with the likelihood based on logistic regression.  Stiefel neural network with $K=2$ ($N=20,604$) provides a variational approximation of the posterior.  The classical simulation focuses on the quantum measurement step -- the computational bottleneck of the algorithm.} 
\label{fig:lower bound}
\end{figure}

Figure \ref{fig:lower bound} plots the lower bound estimates over the iterations. The smooth increasing of the lower bound objective function indicates that the VB algorithm with quantum natural gradient is converging and well behaved.

\section{Conclusion}\label{sec:conclusion}

We have proposed a computationally efficient regression-based approach to estimating the natural gradient for VB, a step towards enabling a greater range of Bayesian inference applications.   We proved that VB with the regression-based natural gradient is guaranteed to converge.  

The regression-based approach supports quantum speed up of natural gradient estimation. Quantum computing is the next-generation computing paradigm that will transform computationally-intensive spheres of academic research and industry.  But quantum computers' power can only apply to problems with specific features, such as those where the handoff of data from a classical to a quantum computer and back can be made economical, or those requiring primarily linear transformation.  We demonstrate that the regression-based approach can make nearly arbitrary VB inference problems suitable for quantum speedup.    

One of the algorithmic tools we propose is the natural gradient readout with quantum Gauss-Southwell rule, a quantum analogue of classical Gauss-Southwell rule.  An intriguing aspect of the quantum Gauss-Southwell method for natural gradient is that it is an example of useful computation infeasible on a classical computer but matter-of-course on a quantum computer.  It provides a taste for novel statistical and data science methods that will co-evolve with the development of quantum computing \cite{lopatnikova.tran:2022:introduction}.

Our work contributes to quantum machine learning: We propose a concrete practical application of quantum linear systems algorithm using a quantum regression method \cite{wiebe2012quantum,schuld2016prediction,duan2017quantum,wang2017quantum,wang2019accelerated,liu2019quantum,kerenidis2020quantum}.      Our method works for a broad class of variational models; it does not require a special-purpose learning model, as in, e.g.~\cite{low2014quantum,rebentrost2014quantum,mcclean2016theory,rebentrost2018quantum,harrow2019low,mcardle2019variational,schuld2019quantum,schuld2020circuit,abbas2020power,park2020geometry} and references therein; quantum speedup for VB with special classes of models include \cite{miyahara2018quantum,zhao2019bayesian}.  Our algorithm may help speed up quantum variational algorithms in in a similar way it helps classical VB.  A broader class of problems could benefit from linearization technique, making them candidates for quantum speedup.

We also contribute to general advances in efficient quantum gradient descent \citep{jordan2005fast,schuld2019evaluating,rebentrost2019quantum,stokes2020quantum,sweke2020stochastic,gacon2021simultaneous,rebentrost2018quantum,mitarai2018quantum,zoufal2019quantum,schuld2020circuit,zoufal2020variational}.  For example, Schuld \emph{et al.}  \cite{schuld2019evaluating} propose the evaluation of the gradient for special-form quantum variational distributions.  We use an alternative approach to avoid encoding a distribution function directly into a quantum state, which can be prohibitively expensive, requiring at least $O(N)$ qubits \cite{rebentrost2019quantum}\footnote{Consider a discretization with $d_1$ points for each dimension.  It would require an $O(h^N)$ dimensional Hilbert space and $O(N\log d_1)$ qubits.}.    

The proposed algorithm adds to the growing list of quantum computing solutions to real-life problems of high practical importance,including, for example, modelling quantum chemistry for development of new materials, medicine, agriculture, and energy (\cite{reiher2017elucidating,aspuru2018matter,cao2018potential,mcardle2020quantum,babbush2018low,macridin2018electron,ajagekar2019quantum} and references therein); modeling biological systems \cite{robert2021resource}; and option pricing \cite{rebentrost2018quantum,stamatopoulos2020option}.

\section*{Appendix}
\begin{proof}[Proof of Theorem \ref{the: main theorem}]
We first need the following lemma:
\begin{lemma}\label{lem: lemma 1}
For every $t\geq 0$, $\E\|Y_t\|\leq \gamma_t C$ with $C>0$ some finite constant. 
\end{lemma}
\begin{proof}[Proof of Lemma \ref{lem: lemma 1}]
We have that
\begin{eqnarray*}
\|Y_0\|&=&\|\widehat{\nabla_\lambda^\text{nat}\tilde{\mathcal L}}(\lambda_0)\|\leq\|\nabla_\lambda^\text{nat}\tilde{\mathcal L}(\lambda_0)\|+\|\epsilon_0\|\\
&\leq&\|I_F^{-1}(\lambda_0)\|\|\nabla\tilde{\mathcal L}(\lambda_0)\|+\|\epsilon_0\|\\
&=& \sigma_{\max}(I_F^{-1}(\lambda_0))\|\nabla\tilde{\mathcal L}(\lambda_0)\|+\|\epsilon_0\|.
\end{eqnarray*}
Hence,
\bean
\E\|Y_0\|\leq \frac{b_1}{b_2}+b_4=C_1\leq \gamma_0C,
\eean
with $C\geq C_1/\gamma_0$.
Assume that $\E\|Y_t\|\leq \gamma_t C$ for some $t\geq 1$. 
Then,
\begin{eqnarray*}
\E\|Y_{t+1}\|&\leq&\zeta_{t+1}\E\|Y_t\|\\
&& +\gamma_{t+1}\E(\|I_F^{-1}(\lambda_{t+1})\|\|\nabla\tilde{\mathcal L}(\lambda_{t+1})\|+\|\epsilon_{t+1}\|)\\
&\leq&\zeta_{t+1}C\gamma_t+\gamma_{t+1}C_1\\
&=&C\gamma_{t+1}(\omega\alpha_t+C_1/C).
\end{eqnarray*}
As $\omega\alpha_t<\omega<1$, with $C\geq C_1/(1-\omega)$ we have $\omega\alpha_t+C_1/C
\leq 1$ for all $t$, and hence,
\[\E\|Y_{t+1}\|\leq C\gamma_{t+1}.
\]
The proof is concluded by reduction with $C=\max\{C_1/\gamma_0,C_1/(1-\omega)\}$.
\end{proof}
We now prove Theorem \ref{the: main theorem}. 
First, 
\begin{align}\label{eq: expression 1}
\langle Y_{t+1},\nabla\tilde{\mathcal L}(\lambda_{t+1})\rangle& =\zeta_{t+1}\langle Y_{t},\nabla\tilde{\mathcal L}(\lambda_{t+1})\rangle \nonumber\\
&+\gamma_{t+1}\langle\nabla_\lambda^\text{nat}\tilde{\mathcal L}(\lambda_{t+1}),\nabla_\lambda \tilde{\mathcal L}(\lambda_{t+1})\rangle\nonumber\\
& +\gamma_{t+1}\langle \epsilon_{t+1},\nabla\tilde{\mathcal L}(\lambda_{t+1})\rangle
\end{align}
By Lemma \ref{lem: lemma 1},
\[\E|\langle Y_{t},\nabla\tilde{\mathcal L}(\lambda_{t+1})\rangle|\leq\E\|Y_{t}\|\|\nabla\tilde{\mathcal L}(\lambda_{t+1})\|\leq b_1C\gamma_t.\]
We have that,
\[\E\langle \epsilon_{t+1},\nabla\tilde{\mathcal L}(\lambda_{t+1})\rangle\leq b_1\E\|\epsilon_{t+1}\|.\]
Taking expectation of \eqref{eq: expression 1} results in
\begin{align}\label{eq: expression 2}
\E\langle Y_{t+1},\nabla\tilde{\mathcal L}(\lambda_{t+1})\rangle&=O(\alpha_{t}^2)+\gamma_{t+1}\E\|\nabla_\lambda \tilde{\mathcal L}(\lambda_{t+1})\|^2_{I_F}\nonumber\\
&+O\big(\gamma_{t+1}\E\|\epsilon_{t+1}\|\big).
\end{align}
In the above, $\|v\|_{I_F}$ denotes the Fisher-Rao metric of vector $v$.
Using Taylor's expansion,
\begin{eqnarray*}
\tilde{\mathcal L}(\lambda_{t+1})&=&\tilde{\mathcal L}(\lambda_{t})-\langle Y_t,\nabla\tilde{\mathcal L}(\lambda_t)\rangle+O(\|Y_t\|^2).
\end{eqnarray*}
By Lemma \ref{lem: lemma 1}, taking the expectation of both sides and summing over $t=0,1,...,n$,
\begin{equation}\label{eq: expression 3}
\sum_{t=0}^n \E\big(\langle Y_t,\nabla\tilde{\mathcal L}(\lambda_t)\rangle\big)=\E\tilde{\mathcal L}(\lambda_{0})-\E\tilde{\mathcal L}(\lambda_{n+1})+\sum_{t=0}^n O(\alpha_t^2).
\end{equation}
Comparing \eqref{eq: expression 2} with \eqref{eq: expression 3}, and noting condition (A3) and that $\sum \alpha_t^2<\infty$, we have that
\[\sum_{t=1}^\infty\gamma_{t}\E\|\nabla_\lambda \tilde{\mathcal L}(\lambda_{t})\|^2_{I_F}<\infty.\]
As
\[\Big(\min_{t=1,...,n}\E\|\nabla_\lambda \tilde{\mathcal L}(\lambda_{t})\|^2_{I_F}\Big)\sum_{t=1}^n\gamma_t\leq\sum_{t=1}^n\gamma_{t}\E\|\nabla_\lambda \tilde{\mathcal L}(\lambda_{t})\|^2_{I_F},\]
taking $n\to\infty$ and because $\sum_{t=1}^n\gamma_t\to\infty$, one must have that
\[\min_{t=1,...,n}\E\|\nabla_\lambda \tilde{\mathcal L}(\lambda_{t})\|^2_{I_F}\to 0,\;\;\;n\to\infty.\]
From Assumption (A2), $\|\nabla_\lambda \tilde{\mathcal L}(\lambda_{t})\|^2_{I_F}\geq \|\nabla_\lambda \tilde{\mathcal L}(\lambda_{t})\|^2/b_3$, hence
\[\min_{t=1,...,n}\E\|\nabla_\lambda \tilde{\mathcal L}(\lambda_{t})\|^2\to 0,\;\;\;n\to\infty,\]
which proves \eqref{eq:result 1}.
To proof \eqref{eq:result 2}, note that
\begin{eqnarray*}
\tilde{\mathcal L}(\lambda^*)&\geq& \tilde{\mathcal L}(\lambda_t)+\nabla \tilde{\mathcal L}(\lambda_t)^\top(\lambda^*-\lambda_t)+\frac{L}{2}\|\lambda^*-\lambda_t\|^2\\
&\geq& \tilde{\mathcal L}(\lambda_t)-\|\nabla \tilde{\mathcal L}(\lambda_t)\|\|\lambda^*-\lambda_t\|+\frac{L}{2}\|\lambda^*-\lambda_t\|^2,
\end{eqnarray*}
where we have used the Cauchy-Schwarz inequality in the second inequality. Then,
\[\|\nabla \tilde{\mathcal L}(\lambda_t)\|\|\lambda^*-\lambda_t\|-\frac{L}{2}\|\lambda^*-\lambda_t\|^2\geq \tilde{\mathcal L}(\lambda_t)-\tilde{\mathcal L}(\lambda^*)\geq 0\]
which implies
\[\|\lambda^*-\lambda_t\|\leq\frac{2}{L}\|\nabla \tilde{\mathcal L}(\lambda_t)\|.\]
Taking expectation, we have
\[\min_{t=1,...,n}\E\|\lambda^*-\lambda_t\|\leq\frac{2}{L}\min_{t=1,...,n}\E\|\nabla \tilde{\mathcal L}(\lambda_t)\|\to 0,\;\;n\to\infty.\]
\end{proof}

\iffalse
\begin{proof}[Proof of Theorem \ref{the: main theorem QNG}]
The proof is similar to that of Theorem \ref{the: main theorem}. 
It is easy to see that Lemma \ref{lem: lemma 1} still holds.
We have that
\[Y_{t+1}=\zeta_{t+1}  Y_t+\gamma_{t+1}(\nabla_\lambda^\text{nat}\tilde{\mathcal L}(\lambda_{t+1})+\epsilon_{t+1}+u_{t+1}),\]
\eqref{eq: expression 1} becomes
\begin{align}\label{eq: expression 1 Them 2}
\langle Y_{t+1},\nabla\tilde{\mathcal L}(\lambda_{t+1})\rangle& =\zeta_{t+1}\langle Y_{t},\nabla\tilde{\mathcal L}(\lambda_{t+1})\rangle \nonumber\\
&+\gamma_{t+1}\langle\nabla_\lambda^\text{nat}\tilde{\mathcal L}(\lambda_{t+1}),\nabla_\lambda \tilde{\mathcal L}(\lambda_{t+1})\rangle\nonumber\\
& +\gamma_{t+1}\langle \epsilon_{t+1},\nabla\tilde{\mathcal L}(\lambda_{t+1})\rangle\nonumber\\
& +\gamma_{t+1}\langle u_{t+1},\nabla\tilde{\mathcal L}(\lambda_{t+1})\rangle.
\end{align}
Taking the expectation of the last term
\begin{align*}
\E\langle u_{t+1},\nabla\tilde{\mathcal L}(\lambda_{t+1})\rangle=\E\Big(\nabla\tilde{\mathcal L}(\lambda_{t+1})^\top\E(u_{t+1}|\mathcal F_{t+1})\Big)=0.
\end{align*}
Hence, we obtain \eqref{eq: expression 2} and the rest of the proof follows from the proof of Theorem \ref{the: main theorem}.
\end{proof}
\fi

%\bibliographystyle{apsrev4-1}
%\bibliographystyle{plain}
\bibliography{QuantumAndMLReferences.bib}

\end{document}